\documentclass[10pt,doublecolumn,conference]{ieeeconf}
\IEEEoverridecommandlockouts % Don't forget this command!

\usepackage{amsfonts}                                                          
\usepackage{epsfig}

\usepackage{amsthm}
\usepackage{graphicx}        
\usepackage{latexsym}
\usepackage{amssymb}
\usepackage{amsmath}
\usepackage{parskip}
\usepackage{multirow}
\usepackage{cite}
\usepackage{mathtools}
\usepackage{epstopdf}
\usepackage{etoolbox}
\usepackage{array}
\usepackage{mathptmx}
\usepackage[colorlinks=false]{hyperref}
\usepackage[bottom]{footmisc}
\usepackage{tikz}

\usetikzlibrary{decorations.pathreplacing}

\usepackage{verbatim}
\usepackage{calrsfs}
\usepackage{booktabs}

\DeclareMathAlphabet{\pazocal}{OMS}{zplm}{m}{n}

\newcommand{\mb}{\mathbb}

\let\bbordermatrix\bordermatrix
\patchcmd{\bbordermatrix}{8.75}{4.75}{}{}
\patchcmd{\bbordermatrix}{\left(}{\left[}{}{}
\patchcmd{\bbordermatrix}{\right)}{\right]}{}{}
\renewcommand{\theequation}{\thesection.\arabic{equation}}
\renewcommand{\thefigure}{\thesection.\arabic{figure}}
\renewcommand{\thetable}{\thesection.\arabic{table}}

\newcommand{\sr}{\stackrel}

\newcommand{\rar}{\rightarrow}

\newcommand{\tri}{\sr{\triangle}{=}}

\newcommand{\be}{\begin{equation}}
\newcommand{\ee}{\end{equation}}
\newcommand{\bea}{\begin{eqnarray}}
\newcommand{\eea}{\end{eqnarray}}
\newcommand{\bes}{\begin{eqnarray*}}
\newcommand{\ees}{\end{eqnarray*}}
\newcommand{\bce}{\begin{center}}
\newcommand{\ece}{\end{center}}
\newcommand{\beae}{\begin{IEEEeqnarray}{rCl}}
\newcommand{\eeae}{\end{IEEEeqnarray}}

\def\VR{\kern-\arraycolsep\strut\vrule &\kern-\arraycolsep}
\def\vr{\kern-\arraycolsep & \kern-\arraycolsep}

\newcommand{\ben}{\begin{enumerate}}
\newcommand{\een}{\end{enumerate}}

\newcommand{\hso}{\hspace{.1in}}
\newcommand{\hst}{\hspace{.2in}}

\newtheorem{theorem}{Theorem}[section]

\newtheorem{remark}{Remark}[section]

\newtheorem{definition}{Definition}[section]
\newtheorem{lemma}{Lemma}[section]

\begin{document}

%\baselineskip=16pt
%\sloppy

%% Paper Title
%% You can use linebreaks \\ within to get better formatting as
%% desired.

\title{On the  Cover and Pombra  Gaussian Feedback Capacity: Complete Sequential Characterizations via a  Sufficient Statistic
  }

%\title{Capacity of Feedback Channels with Memory and Information Lossless Strategies: The Separation Principle of Gaussian Channels and LQG Theory of Directed Information}

%\title{Capacity Achieving Distributions \& Separation  Principle  for Feedback Gaussian Channels with Memory: The LQG Theory of Directed Information}

%\title{Feedback Capacity of Unstable Channels with Memory \& Separation  Principle: The LQG Theory of Directed Information}

%% Author names and affiliations:
%%
%% Avoiding spaces at the end of the author lines is not a problem with
%% conference papers because we don't use \thanks or \IEEEmembership.
%%
%% For several authors with only one affiliation:
%%
% \author{
%   \IEEEauthorblockN{Hui-Ting Chang and Stefan M.~Moser}
%   \IEEEauthorblockA{Department of Electrical and Computer Engineering\\
%     National Chiao Tung University (NCTU)\\
%     Hsinchu, Taiwan\\
%     Email: \{email-of-hui-ting,email-of-stefan\}@ieee.org}
% }
%%
%% For up to three affiliations:
%%

%\author{Charalambos~D.~Charalambous and Christos~K.~Kourtellaris 
%%\thanks{Manuscript received October 9,~2012;~revised....}
%\thanks{This work was financially supported by a medium size University of Cyprus grant entitled ``DIMITRIS".}
%\thanks{The authors are with the Department of Electrical and Computer Engineering, University of Cyprus, 75 Kallipoleos Avenue, P.O. Box 20537, Nicosia, 1678, Cyprus, e-mail: $\{chadcha@ucy.ac.cy,kourtellaris.christos\}$
%}}

%% of the page, use this alternative format:
%%
 \author{
   \IEEEauthorblockN{
     Charalambos D. Charalambous\IEEEauthorrefmark{1} 
     and 
     Christos Kourtellaris\IEEEauthorrefmark{1},
      Stelios Louka\IEEEauthorrefmark{1} 
     \\}
   \IEEEauthorblockA{
     \IEEEauthorrefmark{1}Department of Electrical and Computer Engineering\\University of Cyprus 
%     75 Kallipoleos Avenue, P.O. Box %20537, 
 %    Nicosia, 
  %   1678,
 %     Cyprus
     \\
     Email: chadcha@ucy.ac.cy,  kourtellaris.christos@ucy.ac.cy, louka.stelios@ucy.ac.cy}}

\author{%
Charalambos D. Charalambous, Christos Kourtellaris and Stelios Louka% <-this % stops a space
\thanks{$^{1}$Christos Kourtellaris,  Charlambos D. Charalambous and Stelios Louka are with the Department of Electrical and Computer Engineering, University of Cyprus, Nicosia, Cyprus {\tt\small \{kourtellaris.christos,chadcha, slouka01\}@ucy.ac.cy}}}%
%
%\thanks{*This work was not supported by any organization}% <-this % stops a space
%\thanks{ This work was co-funded by the European Regional Development Fund and the Republic of Cyprus through the Research and Innovation Foundation (Project: EXCELLENCE/1216/0296).}}

\maketitle

\begin{abstract}
  The main objective of this  paper is to derive a new sequential characterization of  the Cover and Pombra   \cite{cover-pombra1989} characterization of the $n-$finite block or transmission  feedback information ($n$-FTFI) capacity, which clarifies  several issues of confusion and incorrect interpretation of   results  in literature.  The optimal channel input processes of the  new equivalent sequential characterizations are  expressed as   functionals of  a sufficient statistic and a Gaussian orthogonal innovations process. From the new representations  follows that the Cover and Pombra  characterization of the  $n-$FTFI   capacity  is  expressed as a functional of  two generalized matrix difference Riccati equations (DRE) of filtering theory of Gaussian systems. This  contradicts results which are redundant  in the literature, and illustrates the fundamental complexity of the feedback capacity formula.  

%The paper contains an  in depth analysis, with examples,   of the specific  technical issues, which are overlooked in past literature  
% \cite{kim2010,liu-han2019,gattami2019,ihara2019,li-elia2019}, that  studied    the AGN channel of  \cite{cover-pombra1989}, for stationary noises.

\end{abstract}

%\newpage
%
%\tableofcontents
%\newpage

\section{Introduction, Motivation, Main Results,   Comparison with Current State  of Knowledge}
\label{sect:problem}
\subsection{The Cover and Pombra Feedback Capacity}
Cover and Pombra \cite{cover-pombra1989} were concerned with the feedback capacity of  the additive Gaussian noise (AGN) channel, 
\begin{align}
Y_t=X_t+V_t,  \ \  t=1, \ldots, n, \ \ \frac{1}{n}  {\bf E} \Big\{\sum_{t=1}^{n} (X_t)^2\Big\} \leq \kappa \label{g_cp_1} 
\end{align}
where   $\kappa \in [0,\infty)$ is the total  power of the transmitter,   $X_t :  \Omega \rar {\mathbb R}$,   $Y_t :  \Omega \rar {\mathbb R}$,  and   $V_t :  \Omega \rar {\mathbb R}$,  are the channel input, channel output and noise random variable, respectively,  and the distribution of the  sequence, $V^n = \{ V_1, \ldots, V_n\}$,  denoted by   ${\bf P}_{V^n}(dv^n)$, is jointly Gaussian, not  necessarily stationary or ergodic.   Cover and Pombra considered,  the set of uniformly distributed messages $W : \Omega \rar  {\cal M}^{(n)} \tri  \left\{1, 2,\ldots, \lceil   2^{nR} \rceil        \right\}$,     codewords  of block length $n$, $X_1=e_1(W),\ldots,  X_n=e_n(W,X^{n-1}, Y^{n-1})$,   and decoder functions,  $y^n \longmapsto d_{n}(y^n)\in  {\cal M}^{(n)}$,  with average  error  probability
\bea
{\bf P}_{error}^{(n)} = {\mathbb P}\Big\{d_{n}(Y^n) \neq W\Big\}= \frac{1}{\lceil 2^{nR }  \rceil } \sum_{W=1}^{\lceil 2^{nR} \rceil } {\mathbb  P}\Big\{d_n(Y^n) \neq W\Big\}. \label{g_cp_4}
\eea
According to  \cite[page 39, above Lemma~5]{cover-pombra1989},  ``$X^n $ is causally  to $V^n$'', which is equivalent to the following decomposition of the joint probability distribution of $(X^n, V^n)$: 
\begin{align}
{\bf P}_{X^n, V^n}
=&{\bf P}_{V_n|V^{n-1}, X^{n}} \; {\bf P}_{X_n|X^{n-1}, V^{n-1}} \; \ldots \; {\bf P}_{V_2|V_1, X^2}   {\bf P}_{V_1|X_1} {\bf P}_{X_1}\nonumber\\
=&{\bf P}_{V^n} \prod_{t=1}^n  {\bf P}_{X_t|X^{t-1}, V^{t-1}} \label{g_cp_3}
\end{align}
That is, ${\bf P}_{V_t|V^{t-1},X^t}={\bf P}_{V_t|V^{t-1}}$, or equivalently $X^t \leftrightarrow V^{t-1} \leftrightarrow V_t$ is a Markov chain, for  $t=1, \ldots, n$. As usual, the messages $W$ are independent of the channel noise $V^n$. 

Cover and Pombra \cite[Theorem~1]{cover-pombra1989}, applied the maximum entropy principle of Gaussian RVs,   derived  direct and converse coding theorem, and   characterized feedback capacity, through 
the $n-$FTFI capacity, given by 
\begin{align}
{C}_{n}^{fb}(\kappa) \tri  \sup_{\substack{{\bf P}_{X_t|X^{t-1},Y^{t-1}}, t=1, \ldots, n: \ \ \ \\ \frac{1}{n}  {\bf E} \big\{\sum_{t=1}^{n} \big(X_t\big)^2\big\} \leq \kappa    }} H(Y^{n})-H(V^n)  \label{ftfic_is_g_def_in}
\end{align}
where $(X^n, Y^n, V^n)$ is jointly Gaussian, provided the supremum exists, and   where $H(\cdot)$ denotes differential entropy. \\
Due to the joint Gaussianity of $(X^n, Y^n, V^n)$, and that, any tuple   of the RVs  $(X^n, Y^n, V^n)$  uniquely  specify the third, the  $n-$FTFI capacity  is given  by \cite[eqn(10)]{cover-pombra1989}
\beae
 C_n^{fb}&(\kappa) 
= \max_{  \big({\bf B}^n, K_{{\bf \overline{Z}}^n}\big):    \frac{1}{n}  trace\big({\bf B}^n \; K_{\bf {V^n}} \; ({\bf B}^n)^T+ K_{\bf {\overline{Z}}^n}\big)  \leq \kappa } \Big\{ \nonumber \\
 & \frac{1}{2} \log \frac{|\big({\bf B}^n+I_{n}\big) K_{{\bf V}^n}\big({\bf B}^{n}+I_{n}\big)^T+ K_{{\bf \overline{Z}}^n}|}{|K_{{\bf V}^n}|} \Big\} \label{cp_12}
  \eeae
where  the channel input process $X^n$ is given by  \cite[eqn(11)]{cover-pombra1989}
\begin{align}
&X_t=\sum_{j=1}^{t-1}{B}_{t,j}V_j +\overline{Z}_t, \hso t=1, \ldots, n, \label{cp_6}   \\
&  
{\bf X}^n ={\bf B}^{n} {\bf V}^n + {\bf \overline{Z}}^n, \hso {\bf Y}^n=\Big({\bf B}^n +I_{n\times n}\Big) {\bf V}^n + {\bf \overline{Z}}^n, \label{cp_10}\\
%&{\bf Y}^n=\Big({\bf B}^n +I_{n}\Big) {\bf V}^n + {\bf \overline{Z}}^n\\
& {\bf \overline{Z}}^n \hso \sim N(0, K_{\bf {\overline{Z}}^n}) \hso  \mbox{and independent of} \hso {\bf V}^n, \label{cp_8} \\
&{\bf X}^n\tri \left[ \begin{array}{cccc} X_1 &X_2 &\ldots &X_n\end{array}\right]^T,  \\
&\frac{1}{n}  {\bf E} \Big\{\sum_{t=1}^{n} (X_t)^2\Big\} =\frac{1}{n}trace \; \Big( {\bf E}\Big({\bf X}^n ({\bf X}^{n})^T\Big) \Big) \leq \kappa. \label{cp_9} 
\end{align} 
and where  notation $N(0, K_{\bf {\overline{Z}}^n})$ means the random variable ${\bf \overline{Z}}^n$ is jointly Gaussian,  with zero mean  and covariance matrix $K_{{\bf \overline{Z}}^n}={\bf E}\{ {\bf \overline{Z}}^n ({\bf \overline{Z}}^n)^T  \}$, and $I_{n}$ denotes the $n\times n $ identity matrix.\\
  The  feedback capacity, $C^{fb}(\kappa)$, is characterized    by  the per unit time limit of the $n-$FTFI capacity  \cite[Theorem~1]{cover-pombra1989}, 
 \begin{align}
 C^{fb}(\kappa) \tri \lim_{n \longrightarrow \infty} \frac{1}{n} C_n^{fb}(\kappa). \label{CP_F}
 \end{align}
 provided the supremum and limit exist. 

\subsection{Main Results of the Paper}
The first main results of this paper is 
\begin{itemize}
\item[R1)] an   equivalent  sequential characterization of the Cover and Pombra \cite{cover-pombra1989} $n-$FTFI capacity, (\ref{cp_12})-(\ref{cp_9}), in which $X^n$ is expressed as,
\end{itemize}
\begin{align}
X_t = \Gamma_t^1 {\bf V}^{t-1} +\Gamma_t^2 {\bf Y}^{t-1} +Z_t,    \hso X_1=Z_1,   \hso t=2, \ldots,n   \label{Q_1_3_s1}
\end{align}
where $Z_t\in  N(0, K_{Z_t}), t=1, \ldots, n$ is an independent , zero mean,   Gaussian sequence, $Z_t \in N(0, K_{Z_t})$ is  independent of $(V^{t-1},X^{t-1},Y^{t-1}),  \hso t=1, \ldots, n$, 
$Z^n$ independent of  $V^{n}$, and $(\Gamma_t^1, \Gamma_t^2) \in (-\infty, \infty) \times (-\infty, \infty) \times$ are   nonrandom, and where $C_n^{fb}(\kappa)$ is 
\begin{align}
{C}_{n}^{fb}(\kappa)
= \sup_{\frac{1}{n}  {\bf E} \big\{\sum_{t=1}^{n} \big(X_t\big)^2\big\} \leq \kappa    }\sum_{t=1}^n \Big\{ H(I_t)-H(\hat{I}_t)\Big\} \label{ftfic_is_g_in}
\end{align}
where $I_t, \hat{I}_t$ are innovations processes defined, 
\begin{align}
I_t \tri Y_t- {\bf E}\big\{Y_t\Big|Y^{t-1}\big\}, \hst  \hat{I}_t \tri V_t- {\bf E}\big\{V_t\Big|V^{t-1}\big\} \label{inn_intr_thm1}
\end{align}
 and where 
the supremum is over all $(\Gamma_t^1, \Gamma_t^2, K_{Z_t}), t=1, \ldots, n$.  The new sequential characterization did not appear elsewhere in the literature.

The second main result is, 
\begin{itemize}
\item[R2)]     the consideration of the the partially observable state space (PO-SS) noise  realization of Definition~\ref{def_nr_2} (below), and an  equivalent  sequential characterization of the Cover and Pombra \cite{cover-pombra1989} $n-$FTFI capacity, (\ref{cp_12})-(\ref{cp_9}), which is expressed as a functional of  two generalized matrix difference Riccati equations (DRE) of filtering theory of Gaussian systems. 
\end{itemize}
 
The third main result is,  
 \begin{itemize}
\item[R3)] the use of a sufficient statistic to expressed ${C}_{n}^{fb}(\kappa)$, which allows the identification of   necessary for the convergence,  $C^{fb}(\kappa) \tri \lim_{n \longrightarrow \infty} \frac{1}{n} C_n^{fb}(\kappa)$, in terms of  properties of  generalized DRE, and its analysis using  sequential methods, such as, dynamic programming. 
\end{itemize}

\begin{definition} 
%A  time-varying PO-SS realization of the Gaussian noise\\
 \label{def_nr_2}
A time-varying PO-SS realization of the Gaussian noise $V^n \in N(0, K_{V^n})$ is defined by  
\begin{align}
&S_{t+1}=A_{t} S_{t}+ B_{t} W_t, \hso t=1, \ldots, n-1\label{real_1a}\\
&V_t= C_t S_{t} +  N_t W_t,  \hso t=1, \ldots, n, \label{real_1_ab}\\
% & \hst = H_t S_{t-1} + D_t W_t,    \hst H_t \tri C_t A_t, \hst D_t\tri C_t B_t+ N_t \label{real_1a}  \\
 & S_1\in N(\mu_{S_1},K_{S_1}), \hso K_{S_1} \succeq 0 \hso \mbox{(positive semidefinite)}, \\
&W_t\in N(0,K_{W_t}),\hso K_{W_t} \succeq 0, \hso  t=1 \ldots, n, \\
& S_t : \Omega \rar {\mathbb R}^{n_s}, \  W_t : \Omega \rar {\mathbb R}^{n_w}, \  V_t : \Omega \rar {\mathbb R}^{n_v}, \\\ &R_t\tri N_t K_{W_t} N_t^T \succ 0, \label{cp_e_ar2_s1_a_new}
%\\
%&  \left[ \begin{array}{cc} Q_t & S_t  \\ S_t & R_t\end{array} \right] \tri cov( \left[ \begin{array}{c} B_tW_t  \\ N_tW_t\end{array} \right], \left[ \begin{array}{c} B_tW_t  \\ N_tW_t\end{array} \right]^T) =  \left[ \begin{array}{cc} B_tK_{W_t}B_t^T & B_t K_{W_t}N_t^T  \\ N_t K_{W_t}B_t^T &N_tK_{W_t}N_t^T\end{array} \right]\ \succeq 0, \hso R_t\tri N_t K_{W_t} N_t^T \succ 0 \label{real_2}
\end{align}
where  $W_t, \  t=1 \ldots, n$ is an independent Gaussian process,  $n_s, n_w$ are arbitrary positive integers, and   $(A_{t}, B_{t}, C_t, N_t,\mu_{S_1}, K_{S_1}, K_{W_t})$ are nonrandom for all $t$, and $n_s, n_w$ are finite positive integers. 
%A time-invariant PO-SS realization of the Gaussian noise $V^n \in N(0, K_{V^n})$ is defined by  (\ref{real_1a})-(\ref{cp_e_ar2_s1_a_new}), with $(A_{t}, B_{t}, C_t, N_t, K_{W_t})=(A, B, C, N,K_{W}), \forall t$. 
\end{definition}

The third main result is,

\begin{itemize} 
\item[R4)] Contrary to the claims in Kim's \cite{kim2010},  the  time-domain characterization of feedback capacity \cite[Theorem~6.1]{kim2010}, does not correspond to the Cover and Pombra code formulation and assumption, and  hence recent literature which makes use of \cite{kim2010}, such as,  \cite{liu-han2017,liu-han2019,gattami2019,ihara2019}, should be read with caution.  
\end{itemize} 
The justification of 4) is easy to verify, by using, for example, the channel input process, considered  in  \cite[Page 76]{kim2010}, for the  autoregressive moving average noise model,  $V_t=c V_{t-1}+ W_t - aW_{t-1},   t=1, \ldots, n$, expressed in state space form, $S_t \tri \frac{cV_{t-1}-aW_{t-1}}{c-a}, t =1, \dots,n$, 
$S_{t+1} = cS_t+W_t, t =1, \dots,n$, $V_t = (c-a)S_t +W_t, \hso t =1, \dots,n$. According to \cite[Page 76]{kim2010},  the channel   input  $X_n$, is 
\begin{align}
&X_1= Z_1, \hst \mbox{$Z_n$ is zero mean, variance  $K_{Z}>0$}
 \label{kim_in_2}\\
&X_n=\Lambda  \Big( S_n-{\bf E}\Big\{S_n \Big|Y^{n-1}\Big\}\Big), \hst n=2,  \ldots  \label{kim_in_3}
\end{align}
However, the above channel input depends on the state $S_n$, and for this to hold it is necessary that the chain of equalities hold:
\begin{align*}
{\bf P}_{X_t|X^{t-1},Y^{t-1}}=&{\bf P}_{X_t|V^{t-1},Y^{t-1}} \; \mbox{always holds by (\ref{g_cp_1})} 
\\
=&{\bf P}_{X_t|V^{t-1},Y^{t-1},S_1} \; \mbox{if   $S_1=s$ is known to the code}\nonumber \\
=&{\bf P}_{X_t|S^t,Y^{t-1},S_1} \; \mbox{if $(V^{t-1}, S_1)$ uniquely defines $S^t$}. 
\end{align*} 
The above shows, 

{\it  a necessary condition for validity of  (\ref{kim_in_2}),  (\ref{kim_in_3}) is: given the initial state of the noise $S_1=s=\frac{cv_{t-1}-aw_{t-1}}{c-a}$, which should be  known to the encoder, the channel noise  $V^n\tri \{X_1, \ldots, V_{n-1}\}$ uniquely defines the state variables $S^n$. }

Clearly, the above necessary condition, never holds  for the noise of Definition~\ref{def_nr_2}, with 
\begin{align*}
&B_tW_t=B_t^1 W_t^1+B_t^2 W_t^2,  \hso N_t W_t=N_t^1 W_t^1 + N_t^2 W_t^2,\\
& \mbox{$W^{1,n}$ and $W^{2,n}$  independent  sequences.}
\end{align*}
We should emphasize that  the above necessary condition is explicitly stated by  Yang, Kavcic, and Tatikonda in \cite{yang-kavcic-tatikonda2007,kim2010}. 

Although, due to space limitations, the complete  proofs of the results of this paper are omitted, these  are found  \cite{charalambous2020new}.

\subsection{Relation to Past Literature}
Over the years,  considerable efforts have been devoted  to compute $C_n^{fb}(\kappa)$ and $C^{fb}(\kappa)$, \cite{yang-kavcic-tatikonda2007,kim2010,liu-han2017,liu-han2019,gattami2019,ihara2019}, often  under simplified assumptions on  the channel noise. In addition,    bounds are described in  \cite{chen-yanaki1999,chen-yanaki2000}, while numerical methods are developed in \cite{ordentlich1994}, mostly for time-invariant AGN channel, driven by stationary noise. We should mention that most papers considered a variant of (\ref{CP_F}), by interchanging the per unit time limit and the maximization operations, 
under the assumption: the joint process $(X^n, Y^n), n=1, 2, \ldots$ is either {\it jointly stationary or asymptotically stationary} (see \cite{kim2010,liu-han2017,liu-han2019,gattami2019}), and  {\it the  joint  distribution of the joint process  $(X^n, Y^n), n=1, 2, \ldots$ is time-invariant}. 
A recent investigation of AGN channels driven by autoregressive unit memory stable and unstable noise with and without feedback is \cite{kourtellaris-charalambous-loyka:2020b,kourtellaris-charalambous-loyka:2020a}, while the connection of ergodic theory and feedback capacity of unstable channels is discussed in \cite{kourtellaris-charalambousIT2015_Part_1,charalambous-kourtellaris-loykaIT2015_Part_2}.

 \section{Equivalent Sequential Characterizations of the Cover and Pombra $n-$FTFI}
  \label{sect:AGN}

\subsection{Notation}
 Throughout the paper, we use the following notation.\\
${\mathbb Z}\tri \{\ldots, -1,0,1,\ldots\}, {\mathbb Z}_+ \tri \{1,\ldots\},    {\mathbb Z}_+^n \tri \{1,2, \ldots, n\}$, where $n$ is a finite positive integer. \\
${\mathbb R}\tri (-\infty, \infty)$,  and ${\mathbb R}^m$ is the vector space of tuples
of the real numbers  for an integer  $n\in {\mathbb Z}_+$.\\
${\mathbb R}^{n \times m}$ is the set of $n$ by $m$ matrices with entries from the set of real numbers for integers   $(n,m)\in {\mathbb Z}_+\times {\mathbb Z}_+$. \\
$\Big(\Omega, {\cal F}, {\mathbb P}\Big)$ denotes a probability space. 
Given a  random variable $X: \Omega \rar {\mathbb R}^{n_x}, n_x\in {\mathbb Z}_+^n$, its  induced    distribution on ${\mathbb R}^{n_x}$ is denoted by  ${\bf P}_{X}$. \\
${\bf P}_{X}\in N(\mu_{X}, K_{X}), K_{X}\succeq 0$ denotes a Gaussian distributed RV $X$, with   mean value  $\mu_{X}={\bf E}\{X\}$ and covariance matrix $K_{X}=cov(X,X)\succeq 0$, defined by  
\begin{align}
K_X = cov(X,X) \tri {\bf E}\Big\{\Big(X-{\bf E}\Big\{X\Big\}\Big) \Big(X-{\bf E}\Big\{X\Big\}\Big)^T \Big\}.\nonumber
\end{align}  
Given another  Gaussian random variables $Y: \Omega \rar {\mathbb R}^{n_y},  n_y\in {\mathbb Z}_+^n$, which is jointly Gaussian distributed with $X$, i.e., the  joint distribution is ${\bf P}_{X,Y}$,   the  conditional covariance of $X$ given $Y$, $K_{X|Y} = cov(X,X\Big|Y)$, is  defined by
\begin{align}
K_{X|Y} \tri & {\bf E}\Big\{\Big(X-{\bf E}\Big\{X\Big|Y\Big\}\Big) \Big(X-{\bf E}\Big\{X\Big|Y\Big\}\Big)^T\Big|Y \Big\}\nonumber\\
=&{\bf E}\Big\{\Big(X-{\bf E}\Big\{X\Big|Y\Big\}\Big) \Big(X-{\bf E}\Big\{X\Big|Y\Big\}\Big)^T \Big\}\nonumber
\end{align}
where the last equality is due to a property of  jointly Gaussian distributed RVs.
%Given three arbitrary RVs $(X, Y, Z)$ with induced distribution ${\bf P}_{X, Y, Z}$, the RVs $(X, Z)$ are called conditionally independent given the RV $Y$ if  ${\bf P}_{Z|X,Y}={\bf P}_{Z|Y}$. This conditional independence is often denoted by,   $X \leftrightarrow Y \leftrightarrow Z$ is a Markov chain.

\subsection{General Equivalent Characterization of Cover and Pombra $n-$FTFI Capacity}
First, we show validity of R1).

\begin{theorem}  
\label{thm_FTFI}
The Cover and Pombra \cite{cover-pombra1989} expressions  (\ref{cp_12})-(\ref{cp_9}), are equivalently represented by (\ref{Q_1_3_s1}), (\ref{ftfic_is_g_in}). 
\end{theorem}
\begin{proof} The complete prove  is found  is given See \cite[Section VI.A]{charalambous2020new}. Below, we provide an outline.  Consider (\ref{cp_6}) and   define the process
\begin{align}
\hso Z_1\tri& \overline{Z}_{1}- {\bf E}\Big\{ \overline{Z}_1 \Big\}, \\ 
Z_t\tri& \overline{Z}_{t}- {\bf E}\Big\{ \overline{Z}_t    \Big|X^{t-1}, V^{t-1}, Y^{t-1}\Big\},\hso t=2, \ldots, n,\\
=&\overline{Z}_t- {\bf E}\Big\{\overline{Z}_{t}\Big| V^{t-1}, Y^{t-1}\Big\}\label{orthogonal_11_n}
\end{align} 
where the last equality is due to, $X^{t-1}$ is uniquely  defined by $(V^{t-1}, Y^{t-1})$.
Then $Z_t$ is a Gaussian orthogonal innovations process, i.e., $Z_t$ is  independent of  $(X^{t-1}, V^{t-1}, Y^{t-1})$, for $t=2, \ldots, n$, and ${\bf E}\big\{Z_t\big\}=0,$ for $t=1, \ldots, n$. 
By (\ref{cp_6}), 
\begin{align}
X_t =&  \sum_{j=1}^{t-1} B_{t,j} V_j +\overline{Z}_t, \hso t=1, \ldots, n,\\
=& \sum_{j=1}^{t-1} B_{t,j} V_j +{\bf E}\Big\{\overline{Z}_{t}\Big| V^{t-1}, Y^{t-1}\Big\}  + Z_t, \hso \mbox{by (\ref{orthogonal_11_n})}\\
\sr{(a)}{=}&  \sum_{j=1}^{t-1} B_{t,j} V_j+ \overline{\Gamma}_t \left( \begin{array}{c}  {\bf V}^{t-1}\\ {\bf Y}^{t-1}\end{array} \right)+ Z_t, \hso \mbox{for some $\overline{\Gamma}_t$}\\
=&  \sum_{j=1}^{t-1} \Gamma_{t,j}^1 V_j+ \sum_{j=1}^{t-1}\Gamma_{t,j}^2   Y_j+ Z_t, \hso \mbox{for some $\Gamma_{\cdot,\cdot}^1, \Gamma_{\cdot,\cdot}^2$}\\
=&\Gamma_{t}^1 {\bf V}^{t-1}+ \Gamma_{t}^2   {\bf Y}^{t-1}+ Z_t, \hst \mbox{by definition} \label{eqn_in_1}
\end{align}
where $(a)$ is due to  the joint Gaussianity of $({Z}^n, X^n, Y^n)$. From (\ref{eqn_in_1}) and  the  independence of $Z_t$ and  $(X^{t-1}, V^{t-1}, Y^{t-1})$, for $t=2, \ldots, n$,  it then  follows (\ref{Q_1_3_s1}) and the  properties.  To show (\ref{ftfic_is_g_in}), we notice that $H(Y^n)=\sum_{t=1}^nH(Y_t|Y^{t-1})= \sum_{t=1}^nH(I_t|Y^{t-1})=\sum_{t=1}^nH(I_t)$ by the orthogonality of the innovations process (\ref{inn_intr_thm1}). Similarly for $H(V^n)$.   
\end{proof}

%\begin{remark} For the  code of Definition~\ref{def_rem:cp_1} that assumes knowledge of the initial state $S_1=s$, it is easy to verify that $C_n^{fb}(\kappa,s)$ is directly obtained from  Theorem~\ref{thm_FTFI}, as a degenerate case (an independent derivation is easily produced  following the derivation of Corollary~\ref{cor_kim}, with slight variations). 
%\end{remark}
%

\begin{remark} By (\ref{ftfic_is_g_in}),  ${C}_{n}^{fb}(\kappa)$ is expressed in terms of the entropy of two invations processes, i.e., independent processes. Consequently, its analysis, such as, the existence of the limit  $\lim_{n \longrightarrow \infty} \frac{1}{n} C_n^{fb}(\kappa)$  is much easier to address, as well its computation. 
\end{remark}

\subsection{A Sufficient Statistic Approach to the  Characterization of $n-$FTFI Capacity of AGN Channels Driven by PO-SS Noise Realizations}
\label{sect_POSS_SS}
Now, we  turn our attention to the derivation of statements under R2) and R3).

We note that  characterization of  the $n-$FTFI capacity 
 $C_n^{fb}(\kappa)$  given  (\ref{ftfic_is_g_in}), although compactly represented,  is   not computationally very practical,  because the input process $X^n$ is not expressed in terms of a {\it sufficient statistic} that summarizes the information of the channel input strategy \cite{kumar-varayia1986}.     \\
We wish to identify a {\it sufficient statistic} for the input process $X_t$, given by (\ref{Q_1_3_s1}),  called the {\it state of the input}, which  summarizes  the  information contained in $(V^{t-1}, Y^{t-1})$. It will then become apparent that the characterization of the $n-$FTFI capacity  can be expressed as a functional of {\it two generalized matrix DREs}.

First, since 
 $C_n^{fb}(\kappa)$ is given by (\ref{ftfic_is_g_in}), we need to compute  the (differential) entropy $H(V^n)$ of $V^n$. The following lemma is useful in this respect.

\begin{lemma} 
\label{lemma_POSS}
Consider the PO-SS realization of $V^n$ of  Definition~\ref{def_nr_2}.  Define the conditional covariance $\Sigma_{t} \tri  cov\big(S_t, S_t\Big|V^{t-1}), \ \Sigma_1 \tri  cov\big(S_1, S_1) =K_{S_1}$ and the conditional mean  of $S_t$ given $V^{t-1}$, $\hat{S}_{t} \tri {\bf E}\Big\{S_t\Big|V^{t-1}\Big\}, \hat{S}_1 \tri \mu_{S_1}$. 
Denote also  the conditional mean and covariance of $V_t$ given $V^{t-1}$, by  $\mu_{V_t|V^{t-1}}, K_{V_t|V^{t-1}}$. \\
The following hold.  \\
(i)  $\hat{S}_t$ satisfies the generalized Kalman-filter recursion   
\begin{align}
&\hat{S}_{t+1}=A_{t} \hat{S}_{t}+ M_{t}(\Sigma_t)  \hat{I}_t, \hso \hat{S}_{1}=\mu_{S_1}, \label{kal_fil_noise} \\
& M_{t}(\Sigma_t) \tri \Big( A_{t}  \Sigma_{t} C_{t}^T+B_{t} K_{W_{t}}N_{t}^T\Big)\Big(N_{t} K_{W_{t}}N_{t}^T+ C_{t} \Sigma_{t} C_{t}^T \Big)^{-1}, \\
&\hat{I}_t \tri V_t -{\bf E}\Big\{V_t\Big|V^{t-1}\Big\}=  V_t - C_t \hat{S}_{t} \nonumber\\ \ & \hspace{0.25cm} =C_t\big(S_{t}- \hat{S}_{t}\big) + N_t W_t, \hso t=1, \ldots, n, \label{inn_po_1} 
\end{align}
\begin{align}
& \hat{I}_t \in N(0, K_{\hat{I}_t}), \hso t=1, \ldots, n \hso \mbox{is an orthogonal innovations}\nonumber \\ & \hspace{0.45cm} \mbox{process, i.e., $\hat{I}_t$ is independent of $\hat{I}_s$, for all $t \neq s$, $\hat{I}_s$,}\nonumber \\
& \hspace{0.3cm}\mbox{ for all $t \neq s$, and $\hat{I}_t$ is independent of  $V^{t-1}$},  \label{inn_po_2}\\
&K_{\hat{I}_t} \tri cov(\hat{I}_t, \hat{I}_t)= C_t \Sigma_t C_t^T +N_t K_{W_t} N_t^T. \label{cov_in_noise}
\end{align}
%(ii) The error $E_t \tri S_t-\hat{S}_t$  satisfies the recursion
%\begin{align}
%E_{t+1} = &M_t^{CL}(\Sigma_t)E_t +  \Big(B_t- M(\Sigma_t) N_t\Big) W_t, \hso E_1= S_1-\hat{S}_1, \hso t=1,\ldots, n,  \label{dre_1_a}  \\
%M_t^{CL}(\Sigma_t)\tri & A_t- M_t(\Sigma_t)C_t.
%\end{align}
(ii) The  covariance of the error, $E_t\tri S_t-\hat{S}_t$  is such that  ${\bf E} \big\{E_t E_t^T\big\}=\Sigma_t$ and satisfies the generalized matrix DRE 
\begin{align}
\Sigma_{t+1}= &A_{t} \Sigma_{t}A_{t}^T  + B_{t}K_{W_{t}}B_{t}^T -\Big(A_{t}  \Sigma_{t}C_{t}^T+B_{t}K_{W_{t}}N_{t}^T  \Big) \nonumber \\
& \hspace{0.1cm} . \Big(N_{t} K_{W_{t}} N_{t}^T+C_{t}  \Sigma_{t} C_{t}^T\Big)^{-1}\Big( A_{t}  \Sigma_{t}C_{t}^T+ B_{t} K_{W_{t}}N_{t}^T  \Big)^T,\nonumber\\   & \hspace{0.1cm} t=1, \ldots, n, \hso \Sigma_{1}=K_{S_1}\succeq 0, \hso \Sigma_t \succeq 0. \label{dre_1}
\end{align}
(iii)  $\mu_{V_t|V^{t-1}}, K_{V_t|V^{t-1}}$ are given by  
\begin{align}
\mu_{V_t|V^{t-1}}=&C_t \hat{S}_{t},  \hso t=1, \ldots, n,  \\
K_{V_t|V^{t-1}}=&K_{\hat{I}_t} =C_t \Sigma_{t} C_t^T + N_t K_{W_t} N_t^T,  \hso t=1, \ldots, n. \label{inn_noise}  
\end{align}
(iv) The entropy of $V^n$, is given by 
\begin{align}
H(V^n)=&\sum_{t=1}^n H(\hat{I}_t)\\
=&   \frac{1}{2}\sum_{t=1}^n \log\Big(2\pi e \Big[C_t \Sigma_{t} C_t^T + N_t K_{W_t} N_t^T\Big] \Big) \label{entr_noise}
\end{align}
\end{lemma}
\begin{proof} This follows from generalized kalman-filter equations  \cite{caines1988}; or \cite[Section II.B, proof of Lemma II.1]{charalambous2020new}
\end{proof}

%On the other hand, for a code that assumes knowledge of the initial state and the state of the noise, and  Conditions 1 and 2 hold,   the characterization of the $n-$FTFI capacity is  expressed as a functional of one generalized DRE (see \cite{yang-kavcic-tatikonda2007}). \\
Next,  we invoke $C_n^{fb}(\kappa)$  given by (\ref{ftfic_is_g_in}) and Lemma~\ref{lemma_POSS} to  show that for each time $t$,  $X_t$ is expressed as
 \begin{align}
&X_t=  \Lambda_{t}\Big(\hat{S}_t- {\bf E}\Big\{\hat{S}_t\Big|Y^{t-1}\Big\}\Big) + Z_t, \hso  t=1, \ldots, n,  \label{suf_s} \\
&\hat{S}_t \tri {\bf E}\Big\{S_t\Big|V^{t-1}\Big\}, \hst \widehat{\hat{S}}_t\tri {\bf E}\Big\{\hat{S}_t\Big|Y^{t-1}\Big\}
\end{align}
which  means, at each time $t$, the  state of the channel input process $X_t$ is $\Big(\hat{S}_t, \widehat{\hat{S}}_t\Big)$; this is the {\it sufficient statistic}. We show that  $\widehat{\hat{S}}_t$ satisfies another  generalized Kalman-filter recursion. \\
Now,  we prepare to prove  (\ref{suf_s}) and the   main theorem.   We start with preliminary calculations.    
\begin{align*}
&{\mb P}\big\{Y_t \in dy \Big| Y^{t-1}, X^t\big\} ={\bf P}_t(dy |X_t,V^{t-1}), \ \ t=2, \ldots, n,  \\
=&{\bf P}_t(dy |X_t,V^{t-1},\hat{S}^{t}), \hso  \mbox{by $\hat{S}_t= {\bf E}\Big\{S_t\Big| V^{t-1}\Big\}$}\\
=&{\bf P}_t(dy |X_t,V^{t-1},\hat{S}_{t}, \hat{I}^{t-1}), \hso  \mbox{by (\ref{inn_po_1}), i.e., $V_t=C_t \hat{S}_t+\hat{I}_t$}\\
=&{\bf P}_t(dy |X_t,\hat{S}_{t}), \hso  \mbox{by $Y_t=X_t+V_t=X_t+C_t\hat{S}_t+\hat{I}_t$ and (\ref{inn_po_2})}.
\end{align*}
At $t=1$,  
${\mb P}\big\{Y_1 \in dy \Big| X_1\big\}={\bf P}_1(dy |X_1)$. By (\ref{PO_state_1}), it follows that the conditional distribution of $Y_t$ given $Y^{t-1}=y^{t-1}$ is 
 \begin{align}
 {\bf P}_t(dy_t|y^{t-1})=&\int {\bf P}_t(dy |x_t,\hat{s}_{t}){\bf P}_t(dx_t|\hat{s}_t,y^{t-1}) {\bf P}_t(d\hat{s}_t|y^{t-1}), \nonumber \\ & \ \ t=2, \ldots, n, \label{PO_tr_1} \\
{\bf P}_1(dy_1)=&\int {\bf P}_1(dy |x_t,\hat{s}_1){\bf P}_1(dx_1|\hat{s}_1) {\bf P}_1(d \hat{s}_1). \label{PO_tr_2}
 \end{align}
From the above distributions,  at each time $t$, the distribution of $X_t$ conditioned on $(V^{t-1}, Y^{t-1})$, induced by (\ref{Q_1_3_s1}), is also expressed as a linear functional of  $(\hat{S}_t, Y^{t-1})$, for $t=1, \ldots, n$. \\
The next theorem further shows that for each $t$,  the dependence of $X_t$ on $Y^{t-1}$ is expressed in terms of  ${\bf E}\Big\{\hat{S}_t\Big|Y^{t-1}\Big\}$ for $t=1,\ldots, n$,  and this dependence gives rise to an equivalent sequential characterization of the Cover and Pombra $n-$FTFI capacity, $C_n^{fb}(\kappa)$.

\begin{theorem} Equivalent characterization of  $n-$FTFI Capacity ${C}_n^{fb}(\kappa)$ for PO-SS Noise realizations   \\ 
\label{thm_SS}
  Consider also the generalized Kalman-filter of Lemma~\ref{lemma_POSS}. \\
Define the conditional covariance and conditional mean  of  $\hat{S}_t$ given $Y^{t-1}$, by 
\begin{align}
K_{t} \tri & cov\Big(\hat{S}_t,\hat{S}_t\Big|Y^{t-1}\Big)= {\bf E}\Big\{\Big(\hat{S}_t- \widehat{\hat{S}}_{t}\Big)\Big(\hat{S}_t-\widehat{\hat{S}}_{t} \Big)^T\Big\}, \\  \widehat{\hat{S}}_{t} \tri& {\bf E}\Big\{\hat{S}_t\Big|Y^{t-1}\Big\},  \hso  t=2, \ldots, n, \label{ric_ga}     \\
\widehat{\hat{S}}_1 \tri &\mu_{S_1}, \hso K_1 \tri 0 \label{ric_gb}.
\end{align}
Then the following hold.  \\
(a) An equivalent characterization of the Cover and Pombra  $n-$FTFI capacity ${C}_n^{fb}(\kappa)$,   defined  by (\ref{cp_12})-(\ref{cp_9}),  is 
\begin{align}
{C}_{n}^{fb}(\kappa)=  \sup_{{\cal P}_{[0,n]}^{\hat{S}}(\kappa)}\sum_{t=1}^n H(Y_t|Y^{t-1})-H(V^n) \label{ftfic_is}
\end{align}
where  $(X^n,Y^n)$ is jointly Gaussian, $H(V^n)$ is the entropy of $V^n$, $\hat{I}^n$ is the innovations process of $V^n$, and 
\begin{align}
&Y_t= X_t +V_t, \hso t=1, \ldots, n, \nonumber\\
&V_t=C_t{\hat{S}}_{t} +\hat{I}_t, \label{cp_16_alt_n}\\
& {\bf P}_t(dy_t|y^{t-1})=\int {\bf P}_t(dy |x_t,\hat{s}_{t}){\bf P}_t(dx_t|\hat{s}_t,y^{t-1}) {\bf P}_t(d\hat{s}_t|y^{t-1}), \nonumber \\ &\hspace{2.3cm} t=2, \ldots, n, \label{PO_tr_1} \\
&{\bf P}_1(dy_1)=\int {\bf P}_1(dy |x_t, \hat{s}_1){\bf P}_1(dx_1|\hat{s}_1) {\bf P}_1(d \hat{s}_1),  \label{PO_tr_2}\\
&{\bf P}_t(dy_t|y^{t-1}) \in N(\mu_{Y_t|Y^{t-1}}, K_{Y_t|Y^{t-1}}),\\
&\mu_{Y_t|Y_{t-1}} \  \mbox{is linear in $Y^{t-1}$ and $ K_{Y_t|Y^{t-1}}$ is nonrandom}, \nonumber \\
&{\bf P}_t(dx_t|\hat{s}_{t}, y^{t-1}) \in N(\mu_{X_t|\hat{S}_{t}, Y^{t-1}}, K_{X_t|\hat{S}_{t}, Y^{t-1}}),  \label{PO_tr_2_b}   \\
&\mu_{X_t|\hat{S}_{t}, Y^{t-1}} \hso \mbox{is linear in $(\hat{S}_{t}, Y^{t-1})$ and nonrandom}, \nonumber  \\
&{\cal P}_{[0,n]}^{\hat{S}}(\kappa) \tri \Big\{{\bf P}_t(dx_t|\hat{s}_{t}, y^{t-1}), t=1,\ldots,n: \nonumber\\ &\hspace{2cm} \frac{1}{n} {\bf E}\Big( \sum_{t=1}^n \big(X_t\big)^2\Big) \leq \kappa    \Big\}. \label{PO_tr_3} 
\end{align}
(b) The optimal jointly Gaussian process  $(X^n,Y^n)$ of part (a)  is  represented, as a function of a sufficient statistic,  by   
\begin{align}
&X_t=  \Lambda_{t}\Big(\hat{S}_{t}- \widehat{\hat{S}}_{t}\Big) + Z_t,\hso t=1, \ldots, n, \label{cp_13_alt_n} \\
%&X_1=  Z_1, \label{cp_14_alt_n}
%\\
&  Z_t \in N(0, K_{Z_t}) \hso \mbox{independent of} \hso (X^{t-1},V^{t-1}, \hat{S}^t, \widehat{\hat{S}^t}, \hat{I}^t, Y^{t-1}), \nonumber \\ & \hspace{0.8cm} t=1, \ldots, n, \nonumber \\
&  \hat{I}_t\in N(0, K_{\hat{I}_t})  \hso \mbox{independent of} \hso (X^{t-1}, V^{t-1},\hat{S}^t, Y^{t-1}, \widehat{\hat{S}^t}), \nonumber \\ & \hspace{0.8cm} t=1, \ldots, n, \nonumber\\
&Y_t= \Lambda_{t}\Big(\hat{S}_{t}- \widehat{\hat{S}}_{t}\Big) + Z_t +V_t, \hso t=1, \ldots, n, \nonumber\\
&\hso \: = \Lambda_{t}\Big(\hat{S}_{t}- \widehat{\hat{S}}_{t}\Big) +C_t{\hat{S}}_{t} +\hat{I}_t  + Z_t,  \label{cp_15_alt_n_1}  \\
%&Y_1=Z_1+V_1, \label{cp_16_alt_n}\\
&\frac{1}{n}  {\bf E} \Big\{\sum_{t=1}^{n} (X_t)^2\Big\}= \frac{1}{n} \sum_{t=1}^{n}\Big(\Lambda_t K_t \Lambda_t^T +K_{Z_t}\Big)  . \label{cp_9_al_n}
\end{align} 
where $\Lambda_t$ is nonrandom.   \\
The conditional mean and covariance, $\widehat{\hat{S}}_t$ and $K_t$, are given by  generalized Kalman-filter equations, as follows.\\
(i)   $\widehat{\hat{S}}_t$ satisfies the Kalman-filter recursion
\begin{align}
&\widehat{\hat{S}}_{t+1}=A_{t} \widehat{\hat{S}}_{t}+F_{t}(\Sigma_{t}, K_t)  I_t, \hso \widehat{\hat{S}}_{1}=\mu_{S_1},  \label{kf_m_1}  \\
&F_{t}(\Sigma_{t}, K_t) \tri \Big(A_{t}  K_{t}\big(\Lambda_t + C_t \big)^T+   M_{t}(\Sigma_{t}) K_{\hat{I}_{t}} \Big)\nonumber\\ &\hspace{1.8cm}\Big\{K_{\hat{I}_t}+   K_{Z_t} + \big(\Lambda_t + C_t \big) K_{t} \big(\Lambda_t + C_t \big)^T \Big\}^{-1}      \\
&I_t \tri Y_t -{\bf E} \Big\{Y_t\Big|Y^{t-1}\Big\}= Y_t-C_t\widehat{\hat{S}}_{t}\nonumber\\ &\hspace{0.2cm}=  \Big(\Lambda_t+C_t\Big) \Big(\hat{S}_t-\widehat{\hat{S}}_{t}\Big)+ \hat{I}_t+ Z_t, \hso t=1, \ldots, n, \label{kf_m_2} \\
&I_t \in  N(0, K_{I_t}), \hso t=1, \ldots, n \hso \mbox{is an orthogonal innovations}\nonumber\\ &\hspace{0.6cm}\mbox{process, i.e., $I_t$ is independent of $I_s$, for all $t \neq s$ }\nonumber \\
&\hspace{0.6cm} \mbox{and ${I}_t$ is independent of  $V^{t-1}$}
\end{align}
\begin{align}
&K_{Y_t|Y^{t-1}}=K_{I_t} \tri cov\big(I_t,I_t\big)\nonumber\\ &\hspace{1.1cm}=  \Big(\Lambda_t +C_t\Big)K_t \Big(\Lambda_t +C_t\Big)^T + K_{\hat{I}_t} + K_{Z_t}, \label{inno_PO}  \\
&K_{\hat{I}_t}  \hso  \mbox{given by  (\ref{cov_in_noise})}. \nonumber
%\label{kf_m_3} 
\end{align}
%(ii) The error $\widehat{E}_t \tri \hat{S}_t- \widehat{\hat{S}}_{t}$ satisfies the recursion 
%\begin{align}
%\widehat{E}_{t+1} = & F_t^{CL}(\Sigma_t, K_t)\widehat{E}_t +\Big(M_t(\Sigma_t)- F_t(\Sigma_t, K_t)\Big) \hat{I}_t-  F_t(\Sigma_t. K_t) Z_t, \hso \widehat{E}_1= \hat{S}_1-\widehat{\hat{S}}_1=0, \hso t=1,\ldots, n,   \label{kf_m_3a}   \\
%F_t^{CL}(\Sigma_t,& K_t)\tri A_t-  F_t(\Sigma_t, K_t)\Big(\Lambda_t+C_t\Big).
%\end{align}
(ii)  $K_t={\bf E}\big\{\widehat{E}_t \widehat{E}_t^T\big\}$ satisfies the generalized DRE 
\begin{align}
K_{t+1}=& A_t K_{t}A_t^T  + M_t(\Sigma_{t})K_{\hat{I}_t}\big(M_t(\Sigma_{t})\big)^T -\Big(A_t  K_{t}\big(\Lambda_t + C_t \big)^T \nonumber\\ &+ M_t(\Sigma_t)K_{\hat{I}_t}   \Big) \Big( K_{\hat{I}_t}+ K_{Z_t} 
+ \big(\Lambda_t + C_t \big) K_{t} \big(\Lambda_t + C_t \big)^T     \Big)^{-1} \nonumber\\ &.\Big(  A_t  K_{t}\big(\Lambda_t + C_t \big)^T+ M_t(\Sigma_t)K_{\hat{I}_t}      \Big)^T, \hso K_t \succeq 0, \nonumber\\& t=1, \ldots, n, \hso K_1=0. \label{kf_m_4_a}  
\end{align}
(c) An equivalent characterization of the $n-$FTFI capacity ${C}_n^{fb}(\kappa)$, defined  by (\ref{cp_12})-(\ref{cp_9}), using the sufficient statistics of part (b),   is  
\begin{align}
 {C}_n^{fb}(\kappa)
%= 
%& \sup_{  \big(\Lambda_{t}, K_{Z_t}\big), t=1, \ldots, n: \hso \frac{1}{n}  {\bf E}\big\{\sum_{t=1}^n \big(X_t\big)^2\big\}\leq \kappa } \frac{1}{2} 
%\sum_{t=1}^n \log \frac{  K_{Y_t|Y^{t-1}}   }{K_{V_t|V^{t-1}}} \label{cp_12_alt_1}\\
= & \sup_{  \big(\Lambda_{t}, K_{Z_t}\big), t=1, \ldots, n: \hso \frac{1}{n}  {\bf E}\big\{\sum_{t=1}^n \big(X_t\big)^2\big\}\leq \kappa } \frac{1}{2} 
\sum_{t=1}^n \log \frac{  K_{I_t}   }{K_{\hat{I}_t}} \\
=& \sup_{  \big(\Lambda_{t}, K_{Z_t}\big), t=1, \ldots, n: \hso  \frac{1}{n}\sum_{t=1}^n\big( \Lambda_t K_t \Lambda_t^T+K_{Z_t} \big)    \leq \kappa } \nonumber\\ &\frac{1}{2} 
\sum_{t=1}^n \log\Big( \frac{ \Big(\Lambda_t +C_t\Big)K_t \Big(\Lambda_t +C_t\Big)^T + K_{\hat{I}_t} + K_{Z_t}     }{K_{\hat{I}_t}}\Big). \label{cp_12_alt_1_new}
  \end{align}
%  where   
%\begin{align}
%&K_{Y_t|Y^{t-1}} =K_{I_t}= (\ref{inno_PO}),   \label{cp_20_alt_1}  \\
%& \frac{1}{n} {\bf E}\Big\{\sum_{t=1}^n (X_t)^2\Big\}= \frac{1}{n}\sum_{t=1}^n\big( \Lambda_t K_t \Lambda_t^T+K_{Z_t} \big), \label{cp_22_alt_1}
%\end{align}
\end{theorem}
\begin{proof} The derivation, although lengthy, is based on the preliminary calculations prior to the statement of the theorem (see \cite[Section VI.D]{charalambous2020new}). 
\end{proof}

The analysis of the per unit time limit  $C^{fb}(\kappa) \tri \lim_{n \longrightarrow \infty} \frac{1}{n} C_n^{fb}(\kappa)$ is carried out in 
\cite{charalambous2020new}.

\begin{remark} 
The characterization of $n-$FTFI capacity  ${C}_n^{fb}(\kappa)$ given by (\ref{cp_12_alt_1_new}), involves the generalized matrix DRE $K_t$ which is also a functional of the generalized matrix DRE $\Sigma_t$ of the error covariance of the state $S^n$ from the noise output $V^n$.  This feature does not appear  in  \cite{kim2010} and recent literature \cite{kim2010,liu-han2019,gattami2019,ihara2019,li-elia2019}, because as explained in R4), 
%because the problems treated by the author are  fundamentally different from the Cover and Pombra formulation.
% \cite{kim2010,liu-han2019,gattami2019,ihara2019,li-elia2019}, because the problems treated by the authors are  fundamentally different from the Cover and Pombra formulation.
\end{remark}

Further to the above remark, if the conditions below hold, \\
A.1) The  feedback code  assumes knowledge of  the  initial state of the noise or the channel, $S_1=s$,  at the encoder and the decoder, and \\
A.2) the noise sequence  $V^{t-1}$ and initial state $S_1=s$  uniquely defines the noise state sequence $S^t$ and vice-versa for $t=1, \ldots, n$, 

then  in Theorem~\ref{thm_SS}, $X_t$ is reduced to   $X_t =\Lambda_t \Big(S_t -{\bf E}\Big\{S_t\Big|Y^{t-1},s_1\Big\}+Z_t, t=1,\ldots, n$, and all equations are  simplified, precisely as in  Yang, Kanvic and Tatikonda \cite{yang-kavcic-tatikonda2007}.

\label{sect_POSS}

\section{Conclusion}
New equivalent sequential characterizations of the  Cover and Pombra \cite{cover-pombra1989}  ``$n-$block'' feedback capacity  formulas are derived using time-domain methods,  for  additive Gaussian noise (AGN) channels driven by nonstationary  Gaussian noise. The new feature of the equivalent characterizations are the representation of  the  optimal channel input process by a {\it sufficient statistic and  Gaussian orthogonal innovations process}. The sequential characterizations of the $n-$block feedback capacity formula are expressed  as a functional of   two generalized matrix difference Riccati equations (DRE) of filtering theory of Gaussian systems. 

\section{Acknowledgements}
This work was supported in parts by  the European Regional Development Fund and
 the Republic of Cyprus through the Research Promotion Foundation Projects EXCELLENCE/1216/0365 
and  EXCELLENCE/1216/0296

\newpage 

\bibliographystyle{IEEEtran}

\bibliography{Bibliography_capacity}

\end{document}